\documentclass[article,aps,pra,onecolumn,groupedaddress]{revtex4}
\usepackage{amsmath}    
\usepackage{graphicx}   
\usepackage{epsfig}
\newtheorem{theorem}{Theorem}
\newtheorem{lemma}[theorem]{Lemma}

\newenvironment{proof}[1][Proof]{\noindent\textbf{#1.} }{\ \rule{0.5em}{0.5em}}
\begin{document}

\title{Independent Attacks in Imperfect Settings; A Case for A Two Way Quantum Key Distribution Scheme}
\author{J.S. Shaari$^{a}$ and Iskandar Bahari$^{b}$}
\address{\smallskip $^{a}${\textit{Faculty of Science, International Islamic University of Malaysia (IIUM), Jalan Bukit Istana, 25200 Kuantan, Pahang, Malaysia}} 
\\
$^{b}${\textit{Information Security Cluster, MIMOS Berhad, Technology Park Malaysia, 57000 Kuala Lumpur, Malaysia}}
}
\begin{abstract}
We review the study on a two way quantum key distribution protocol given imperfect settings through a simple analysis of a toy model and show that it can outperform a BB84 setup. We provide the sufficient condition for this as a ratio of optimal intensities for the protocols.
\end{abstract}

\maketitle
\section{Introduction}

\noindent The study of quantum key distribution (QKD) \cite{gisin} has seen much development since its debut in the seminal work of Bennett and Brassard (BB84)\cite{bb84}, where a new framework in secure communications based on physical laws was introduced. Security analysis which was initially  more confined to studies within a theoretical framework was later extended to consider imperfections in a realistic setup \cite{lutk}. While the former sees theoretical challenges where the (sometimes perfect) legitimate users, Alice and Bob are pit against an adversary, Eve who has perfect technological advantage, the latter addresses imperfections of the users, e.g. using weak coherent pulses instead of a single photon source (which immediately opens them to attacks like the notorious photon number splitting attacks (PNS) \cite{lutj,hut} given lossy channels). Limitations on Eve is also not uncommonly studied, e.g. in \cite{lutk} where Eve is limited to independent attacks and \cite{bech} where she is robbed of a quantum memory. 

While prepare and measure QKD schemes like BB84 are exhaustively studied, other families of QKD have received less treatment partly due to their more complicated nature. One such family is the `two way QKD schemes' which arguably began with the Ping Pong protocol \cite{bostroem} and was later followed by nonentangled versions reported in \cite{cai,deng,LM,LM05}. Essentially all shared the particular feature where Bob would send a qubit to Alice (forward path) who would encode using a flip (passive) operator to flip (retain) the state received. The qubit would then be returned to Bob (backward path) who would make a sharp measurement to deduce Alice's operation. As Alice's operation defines the encoding of the qubits, an Eve wishing to glean information must necessarily attack both the forward as well as the backward path. Security is ensured by virtue of a control mode where Alice would randomly make a measurement instead and results would be compared on a public channel later to ascertain errors. Details of this may be found in \cite{deng,LM}. It is important to note that complete security analysis on such protocols has never been done though available calculations seem to suggest a higher level of robustness compared to BB84. Recently, a study on cloning unitary transformations in \cite{chiri} seems to support this case. For the purpose of being specific, we will refer more often than not to LM05, which was the name given in \cite{LM05} to the protocol \cite{LM}. Like BB84, a treatment for LM05 in terms of an imperfect source was done in \cite{marco}. These works have suggested the robustness of LM05 over BB84 against PNS \footnote{detailed discussions along this line were in private communications with Lucamarini}. The results then \cite{marco} exhibited this at least for certain short and medium distances. However we feel there are some pertinent issues related to this protocol that have yet to receive rigorous highlight of which is our intent in this letter.

The outline of our work is as follows. We first propose a toy model for such a two way protocol; i.e. a model protocol which essentially mimics the LM05. The model in some sense would distill only the most essential features of a two way QKD scheme and is inherently simpler to analyze. We should emphasize the point, that we \textbf{are not} proposing a new protocol. We then adopt the formalism for the optimal attack as described in \cite{bech} where Eve interacts with Alice/ Bob's qubit using a two dimensional ancilla and measures independently in the forward as well as the backward path. Let us note that such an attack is quite sufficient in terms of an individual attack on the BB84. Subsequently, we prove a simple theorem where the strength of interaction (attack) in the backward path should be equal to the one in the forward path for Eve's benefit. We note that this attack is partially inspired by the `optimal incoherent attack' in \cite{LM}. 
We proceed to consider the case for a non-ideal Alice/ Bob where they operate with a lossy channel as well as an imperfect photon source. The main objective of the toy model is to show that, in the face of these imperfections, its proper merit (and more importantly that of two way QKD) is really in its two way nature rather than Bob's `deterministic' measurement\footnote{by deterministic we refer to Bob's measurement basis which always coincide with the basis the qubit was prepared in \cite{jes}}. We present the secure key rate which interestingly enough outperforms BB84 at all distances. 

\section{A Two Way QKD Scheme; A Toy Model}
\noindent Let us first describe our model. We imagine a protocol, like LM05 where Bob sends to Alice a qubit prepared in one of four states $\left|i\pm\right\rangle$ where $i=x,y$ of two preferred basis ($Y$ and $X$). Alice then measures in either of the two basis chosen randomly; which projects it to an eigenstate of her measurement operator. The measured qubit, say $\left|j\pm\right\rangle$, where $j=x,y$ then is subjected to a unitary transformation, $I$ or $Z$ as follows;
\begin{eqnarray}
I\left|j\pm\right\rangle\rightarrow\left|j\pm\right\rangle~~~,~~~~Z\left|j\pm\right\rangle\rightarrow\left|j\mp\right\rangle
\end{eqnarray}
to retain or to flip the qubit respectively before resending to Bob. Bob then proceeds to measure in the basis he originally prepared the qubit in. At the end, they should reveal the basis over a public channel and only the cases where they share the same basis would Alice and Bob share the information of what unitary transformation was carried out; i.e. $\mathcal{A}+\mathcal{B}\mod{2}$ where  $\mathcal{A}$ and $\mathcal{B}$ are the bit values resulting from Alice's and Bob's measurements. 

The model protocol is essentially LM05 with the added feature that Alice always measure before her unitary transformation and at first glance may seem simply as a two way channel derived from two BB84 `put back-to-back'. However it is important to note that the encoding is really derived from the sum of the bit content obtained from Bob's and Alice's stations. This fact forces Eve to attack both paths. One may argue that the transformation is rather spurious in that Alice could very well just prepare a state identical to her measurement outcome (or an orthogonal one) to resend; in fact in realistic situations, a qubit (usually a photon) is absorbed by a detector in a measurement. Nevertheless, we prefer to retain this model for the sake of having the projected qubit treated on. We will note later the merits of such an assertion. Another point worth mentioning is Alice's measurements in a certain sense allows for a detection of Eve, identical to the control mode in LM05 and does not require us to view errors in control mode as opposed to encoding mode separately (which is necessarily the case for LM05; the double CNOT attack for example in LM05 would induce errors in the control mode only \cite{arxivrev}). In terms of efficiency though, it is immediately half that of LM05 due to the random choice of measurement basis. This should be analogous to having 50\% control mode in LM05.   

\subsection{Independent Attack; A Two Dimensional Ancilla}
\noindent We now analyze the model (for the sake of brevity, we hereafter refer to as ToM), under an independent attack using the formalism based on \cite{bech} where Eve introduces her ancilla to interact with Alice/ Bob's qubit and measures thereafter. This is what is referred to as the `optimal' attack in \cite{bech}, where the interaction between her ancilla and the travelling qubit be written in the $Z$ basis as 
\begin{eqnarray}
\left|0\right\rangle\left|0_E\right\rangle\rightarrow \left|0\right\rangle\left|0_{E}\right\rangle\\\nonumber
\left|1\right\rangle\left|0_E\right\rangle\rightarrow \cos{\alpha}\left|1\right\rangle\left|10_{E}\right\rangle+\sin{\alpha}\left|0\right\rangle\left|01_{E}\right\rangle
\end{eqnarray}
where we consider $\alpha\in [0,\pi/2]$. The formalism considers Alice and Bob using the $X$ and $Y$ bases, thus immediately ensuring that the errors denoted by Alice and Bob would be the same in both bases and the state fidelities for Alice/ Bob and Eve would be $\left(1+\cos{\alpha}\right)/2$ and $\left(1+\sin{\alpha}\right)/2$ respectively \cite{bech}. These fidelities quoted are true since Eve has knowledge of basis (akin to a BB84 basis revelation). This is one of the niceties of our toy model as it allows for a direct import of such an attack's formalism. Without a measurement made by Alice, the qubit would be in an entangled state and further operations would be very messy and becomes somewhat complicated. In consideration of ToM, Eve attacks in the forward path by allowing her ancilla to interact with the qubit sent by Bob to Alice.  A second attack is launched on the backward path assuming a fresh ancilla with similar unitary action (except that an angle $\beta$ is used instead). It is easy to see, as noted in \cite{LM} that Eve would be able to make the correct guess of Alice's encoding only in two possible instances; when she guesses the states in the forward and backward path correctly as well as when she guesses wrongly in both the paths. Thus we note that the fidelity of guessing the transformation correctly is given by 
\begin{eqnarray}
F_{E}&=&\dfrac{\left(1+\sin{\alpha}\right)}{2}\dfrac{\left(1+\sin{\beta}\right)}{2}+\dfrac{\left(1-\sin{\alpha}\right)}{2}\dfrac{\left(1-\sin{\beta}\right)}{2}\\\nonumber
&=&\dfrac{\left(1+\sin{\alpha}\sin{\beta}\right)}{2}.
\end{eqnarray}
In consideration of the information shared between Alice and Bob, a bit is shared only with probability
\begin{eqnarray}
F_{AB}&=&\dfrac{\left(1+\cos{\alpha}\right)}{2}\dfrac{\left(1+\cos{\beta}\right)}{2}+\dfrac{\left(1-\cos{\alpha}\right)}{2}\dfrac{\left(1-\cos{\beta}\right)}{2}
\\\nonumber
&=&\dfrac{\left(1+\cos{\alpha}\cos{\beta}\right)}{2}
\end{eqnarray}  
and the probability for an erroneous bit would be $1-F_{AB}$. 
A proper choice for the pair $\left(\alpha,\beta\right)$ should be made to ensure that Eve's fidelity is maximized for any given disturbance experienced by Bob. We propose the following; given the `independent ancilla based attack' as defined above, Eve achieves the highest fidelity when  $\alpha=\beta$. In order to prove this, we begin with the following lemma.

\begin{lemma}
For any pair  $\left(\alpha,\beta\right),\alpha,\beta\in[0,\pi/2]$, $\exists \Phi$ such that $\cos{\alpha}\cos{\beta}=\cos^2{\Phi}$
\end{lemma}
\begin{proof}
As $\alpha,\beta\in[0,\pi/2]$, therefore $0\leq \cos{\alpha}\cos{\beta}\leq 1$. It becomes immediate to see one may solve the following equality $\sqrt{\cos{\alpha}\cos{\beta}}-\cos{\Phi}=0$ for $\Phi\in[0,\pi/2]$.
\end{proof}
\begin{theorem}
Given a two way QKD protocol (ToM), an independent ancilla based attack using a two dimensional ancilla sees Eve achieving the highest fidelity when  $\alpha=\beta$ $\forall \alpha,\beta\in[0,\pi/2]$.
\end{theorem}
\begin{proof}
Starting with $\cos{\alpha}\cos{\beta}=\cos^2{\Phi}$, Bob's disturbance may be written as 
\begin{eqnarray}
\dfrac{\left(1-\cos{\alpha}\cos{\beta}\right)}{2}=\dfrac{\left(1-\cos^2{\Phi}\right)}{2}=\dfrac{\sin^2{\Phi}}{2}
\end{eqnarray}
and Eve needs to find a pair $(\alpha,\beta)$ that would maximize her fidelity for a given $\Phi$.
Writing $\cos{\alpha}\cos{\beta}=\cos{\left(\alpha-\beta\right)}-\sin{\alpha}\sin{\beta}$ we arrive at the following
\begin{eqnarray}
\dfrac{\sin^2{\Phi}}{2}&=&\dfrac{1-\cos{\left(\alpha-\beta\right)}+\sin{\alpha}\sin{\beta}}{2}\\\nonumber
\sin^2{\Phi}&\geq &\sin{\alpha}\sin{\beta}.
\end{eqnarray}
The above inequality is valid as $1-\cos{\left(\alpha-\beta\right)}\geq 0$ and an equality is achieved when $\alpha=\beta$ and $\alpha=\Phi$. As Eve's fidelity function, 
$F_{E}\left(\alpha,\beta\right)$ is an increasing function of $\sin{\alpha}\sin{\beta}$, $\max\left[F_{E}\left(\alpha,\beta\right)\right]=F_{E}\left(\alpha=\beta\right).$
\end{proof}
\newline
\newline
\noindent In the ensuing discussion, we consider $\alpha=\beta$ and $F_{E}$ and $F_{AB}$ would reduce to   
$F_{E}=\left(1+\sin^2{\alpha}\right)/2$
and $F_{AB}=\left(1+\cos^2{\alpha}\right)/2$ respectively.
We plot below in Figure \ref{fig:f1}, the information curves for Alice-Bob, $I_{AB}$, and Alice-Eve, $I_{ToM}$ under this ancilla based attack. We also include the case where Eve uses the simple intercept resend (IR) attack and as ToM contains a basis revelation step, the IR would be analogous to BB84. In an IR, we may assume that Eve attacks only a fraction, say $x$ of the qubits and her information of Alice's encoding for the attacked fraction would be complete. We see that this is in fact the best strategy for Eve. For comparison purposes, we include the corresponding curve for an optimal BB84 attack. 
\begin{figure}[h]
\center
			\includegraphics[angle=270,width=0.6\textwidth]{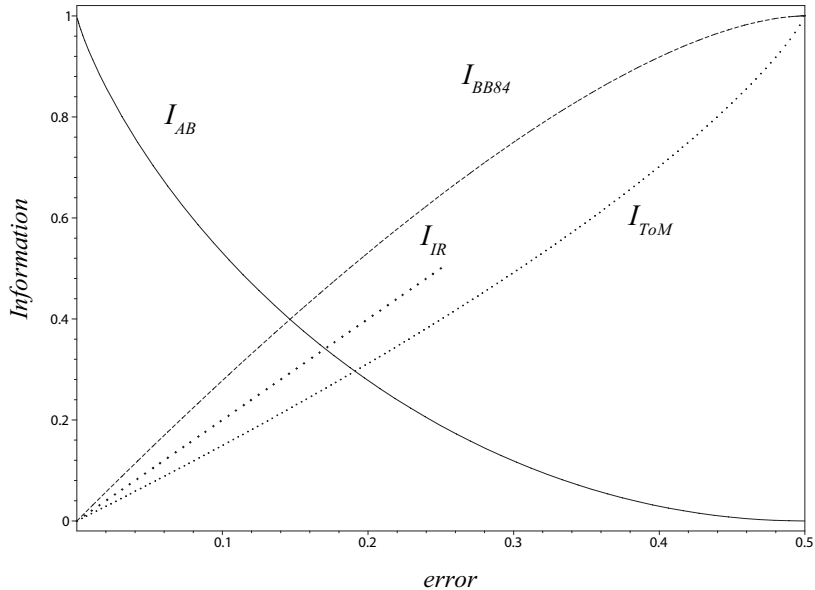}
			\caption{The above shows the information gained by Eve for different attack strategies employed. $I_{ToM}$ is the Alice-Eve mutual information in an ancilla base attack as described in text, $I_{IR}$ is Eve's information using the IR attack and $I_{BB84}$ is Eve's information for an optimal independent eavesdropping in BB84.\label{fig:f1}}
\end{figure}

\section{The Imperfect Source}
\noindent Rather than a perfect single photon source, what Alice and Bob normally employ is in fact a source of which pulses emitted may contain more than one photon (or even zero) and the distribution is given by the Poissonian statistics. The probability to have $n$ photons in a pulse is given by $P_n=\mu^n \exp{-\mu}/n!$. The critical point in this drawback is that it opens Alice and Bob to an attack such as the photon number splitting (PNS) attack \cite{lutj,hut} where Eve may steal a certain number of photons from a multiphoton pulse and make measurements on them while not disturbing the photons measured by Bob (thus not introducing any errors). In a BB84 setup, Lutkenhaus \cite{lutk} has shown that, given an error rate $e$, a secure key can only be generated from single photon contributions at the rate 
\begin{eqnarray}
G=\dfrac{1}{2}p_{b}\left[-f(e)h(e)+\beta_b \left(1-\tau\left(\dfrac{e}{\beta_{b}}\right)\right)\right]
\label{key}
\end{eqnarray} 
where $f(e)$ reflects the error correction efficiency (we take as 1.22 for simplicity) while $h(e)$ is the binary entropic function. The average probability for detection at Bob's is denoted by $p_{b}$ and $\beta_{b}$ essentially reflects the fraction of bits from which a secure key may be extracted. The term $\tau$ is a function for the amount of bits discarded due to privacy amplification and is defined as $\tau (e)=\log_2{1+4e-4e^2}$ for $e<1/2$ and 1 for $e\geq 1/2$.

\subsection{ToM and Imperfect Source}
\noindent In order to consider the study of ToM given an imperfect setup, it is important to highlight our assumptions regarding Alice;
\begin{enumerate}
	\item  Alice \textit{can} actually measure a state and retain the projected state as a photon which she may subject to the next procedure. 
	\item  We further assume that Alice's measurement and instruments are completely efficient in the sense that the transmitivity of her instruments is really unity. 
\end{enumerate}
Despite the fact that the two assumptions above are not realistic given today's technology; we should reiterate that we are \textbf{not} testing a new protocol. ToM is after all, only a toy model with which we hope to highlight the use of the 2 way channel for protocols like LM05. Hence, with regards to the first assumption, we are only interested in the case where Eve would exploit the imperfection of the photon source as well as the lossy channel; other than that we assume Alice's technological \textit{fantasies}. The second assumption lies in the comfort of the fact that in LM05, the key rate does not suffer any imperfection of Alice's measurement in an encoding mode. However for the sake of the analysis of an imperfect source, we insist that Alice ignores/ cannot determine the number of photons in a pulse and in the multiphoton case, her measurement operator should act on all qubits in the pulse. This consequently presents a more pessimistic scenario for ToM.
 
Lucamarini et al. \cite{marco} gave a formula for secure key rate for LM05  similar to eq.(\ref{key}) except for the absence of the `$1/2$' term as well as a different `$\beta$' (which includes double photon contributions). The arguments for a two photon contribution in \cite{marco} carry over to ToM quite straightforwardly. However in ToM, while Eve's preferable attack would be the simple IR, the case for a two photon source is different. Eve could very well attack only the backward photon (subsequent to Alice's transformation) after retaining one of photons in the forward path. Once the basis is revealed publicly, she may make a sharp measurement on the hijacked photon and thus her fidelity of Alice's transformation would be perfectly identical to BB84's, $\left(1+\sin{\alpha}\right)/2$ (a similar though then a heuristic justification for formulas used for LM05 was made in \cite{jis}). 
As for 3 photon pulses, unlike \cite{marco}, basis revelation in ToM always allows for conclusive measurements for such pulses. Thus, with $p_t$ as the signal detected at Bob's station, the fraction from which to distill a secure key for ToM is given by
\begin{eqnarray}
\beta_{t}=\dfrac{1}{p_t}\left[p_t-\left(1-e^{-\mu}\sum_{i=0}^2{\dfrac{\mu^i}{i!}}\right)\right].
\end{eqnarray}

\noindent In order to determine how much information is to be discarded in privacy amplification, we need to ascertain the amount of Renyi information Eve may have access to. Let us note in passing that \cite{marco} used the function $\tau$ as defined in \cite{lutk} which we believe to be a very pessimistic estimate.

As a quick refresher, we note that the Renyi entropy of order 2 for a random variable with $n$ outcomes with $p_i$ being the probability for \textit{i-th} outcome is given by \cite{renyi} 
\begin{eqnarray}
H_R=-\log_2{\sum_{i=1}^{n}p_i^2}
\end{eqnarray}
and the Renyi information gain may be given by the difference between the \textit{apriori} and \textit{aposteriori} entropies \cite{how}. 
In a simple IR attack where Eve attacks only a fraction $x$ of the qubits, her \textit{aposteriori} Renyi entropy of Alice's encoding is given by 
\begin{eqnarray}
-\log_2\left[\left(\dfrac{1}{2^2}+\dfrac{1}{2^2}\right)^{1-x}\right]-\dfrac{1}{2}\log_2\left[\left(\dfrac{1}{2^2}+\dfrac{1}{2^2}\right)^{x}\right]-\dfrac{1}{2}\log_2\left[\left(1^2+0^2\right)^{x}\right]\\\nonumber
=-\log_2\left[2^{\left(x/2\right)-1}\right]
\end{eqnarray} 
and her Renyi information gain (with $0.25x=e$)
\begin{eqnarray}
\tau_{t}(e)=1+\log_2\left[2^{\left(x/2\right)-1}\right]=2e
\label{lm05}
\end{eqnarray} 
The Renyi information gain considering Eve's fidelity in the case of a two photon pulse is given by 
\begin{eqnarray}
1+\log_2{\left\{[(1+\sin{\alpha})/2]^2+[(1-\sin{\alpha})/2]^2\right\}}.
\end{eqnarray} 
With the disturbance $e=(1-\cos{\alpha})/2$, the amount of bits to be discarded in privacy amplification is thus $\log_2{\left(1+4e-4e^2\right)}$, which is really the $\tau$ from \cite{lutk}.    
As $\tau(e)>\tau_{t}(e),\forall e\in\left(0,0.5\right)$ (note that $\tau_{t}(e)$ is defined on $e$ only up to 0.25)   and since Alice and Bob cannot really ascertain where the errors are from, a safe choice for the amount of bits to be discarded in privacy amplification would be given by $\tau$. Thus the use of $\tau$ is aptly and rigorously justified for calculating ToM's secure key rate.

We proceed to consider in the following subsection, the secure key rate formula as a function of distance for ToM in a fiber based implementation.

\subsection{Secure Key Rates}
\noindent A fiber based setup has received fair treatment in modeling \cite{lutk,marco,X} and we will follow mostly \cite{X} where the overall gain for an encoded qubit/ photon detected by Bob's measurement as well as the overall QBER would be given respectively by   
\begin{eqnarray}
p=p^{dark}+1-e^{-\mu\eta t}
\label{xp}
\end{eqnarray} 
and
\begin{eqnarray}
e=\dfrac{e_0p^{dark}+e_{det}\left(1-e^{-\mu\eta t}\right)}{p}
\label{qber}
\end{eqnarray}
where $\eta$ includes the transmitivity of internal optical components as well as the efficiency of Bob's detector. The transmitivity, $t=10^{-\gamma l_k}$ where $\gamma$ is coefficient value for fiber loss and $l_k$ is the total distance a photon travels in the fiber ($l_{BB84}=2l_{ToM}$)\cite{marco}. The $p^{dark}$ term reflects the inevitable contribution from the dark counts while $e_0$ represents the error stemming from the dark counts and is given as $0.5$. The optical system's alignment and stability is characterized by $e_{det}$. 

Following the lead of \cite{marco,arxivrev} and given equations (\ref{key},\ref{xp},\ref{qber}), our formula for key rate has the form identical to eq.(\ref{key})
except for the $p_b$ substituted with  $p_t$ and $\beta_b$ with $\beta_t$. As the agreed bit occurs only half the time, our comparison puts BB84 and ToM somewhat on equal footing. We plot in Figure \ref{fig:f2} the secure key rates for ToM and compare them to the estimates for a BB84 with the experimental parameters of the GYS \cite{X,gys} as well as the KTH \cite{X,kth} experiments. The term $p_b$ and error for BB84 is also calculated using equations (12) and (13).

\begin{figure}[hbtp]
	\centering
		\includegraphics[angle=270,width=0.65\textwidth]{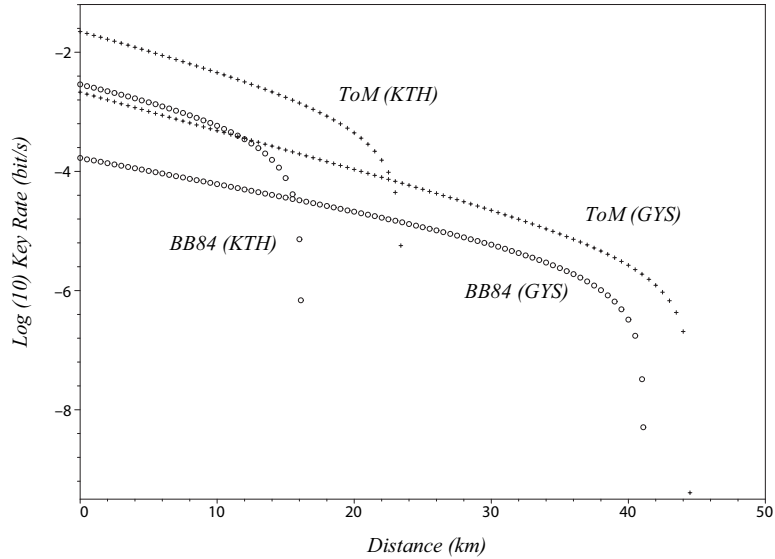}
		\caption{The above exhibits the secure key gain for ToM and BB84 for the parameters based on GYS and KTH. A numerically optimized $\mu$ at each distance for each protocol has been used for the distances plotted.\label{fig:f2}}
\end{figure}

It is thought that the toll of lossy channels on two way schemes should leave its performance somewhat wanting compared to BB84. However, our consideration for ToM here illustrates how it exceeds the key rate as well as performance distance of BB84 (we refer to Figure \ref{fig:f2}). This may be understood as follows. If we consider purely the detection at Bob's station, for equal source intensity, $\mu$, then obviously $p_b>p_{t}$. However, it is critical to remember from \cite{lutk}, the choice for source intensity should be optimal to achieve a maximum key rate for every distance. Writing the optimal intensity for BB84 and ToM as $\mu_{opt}$ and $k\mu_{opt}$ respectively, where $k$ is a constant for a given distant, then $p_b<p_{t}$ when 
\begin{eqnarray}
\exp\left({-\eta\mu_{opt}10^{-\frac{\gamma l}{10}}}\right)&>&\exp\left({-\eta k\mu_{opt}10^{-\frac{\gamma 2l}{10}}}\right)\\\nonumber
k&>&10^{\frac{\gamma l}{10}}.
\label{ine}
\end{eqnarray}
Following the above inequality, it is obvious that $\beta_b<\beta_t$ and writing the QBER of eq.(\ref{qber}) as $e_{edet}+\left(e_0-e_{det}\right)p^{dark}/p$; it follows then that $e_b>e_t$ ($e_b$ and $e_t$ are the QBERs for BB84 and ToM respectively). Hence the inequality (15) can be seen as a sufficient condition for the key rate of ToM to be greater than that of BB84. As an example, for GYS, we observe that at $\gamma= 0.21$ and $l=41$ km, $k>7.26$. The ratio of optimal intensity for ToM to BB84 at this distance is about $9$. Another example is for KTH parameters at 16 km, $k>2.1$ while the ratio observed is $4.6$.  

In the above calculations we have actually considered a perfectly efficient Alice. We now consider briefly the case when Alice's transmitivity, $\eta_A<1$. The sufficient condition would now be corrected to $k>\eta_A^{-1}10^{\frac{\gamma l}{10}}$. We plot in Figure \ref{fig:f3} and Figure \ref{fig:f4} several curves for varying transmitivity of Alice for GYS and KTH parameters respectively. 
\begin{figure}[h!]
	\centering
			\includegraphics[angle=270,width=0.6\textwidth]{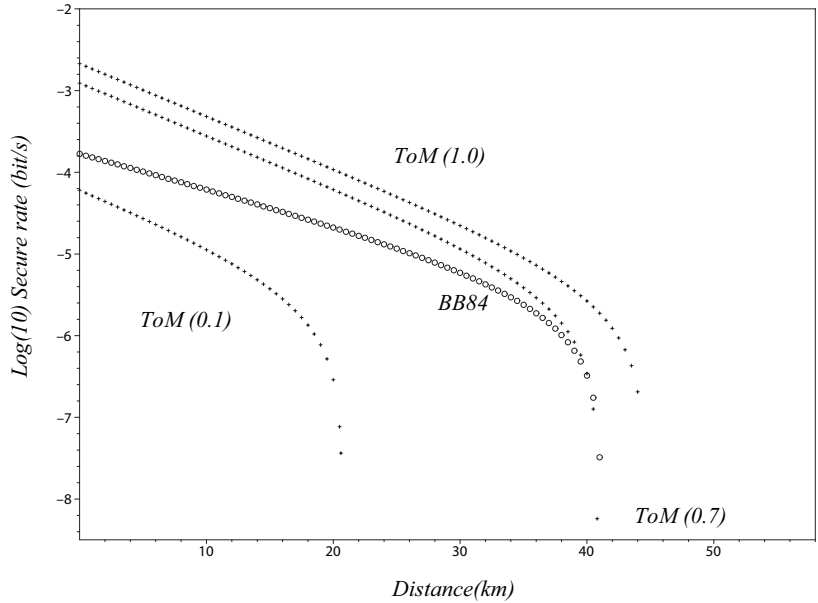}
		\label{3D}
	\caption{The above exhibits the secure key gain for ToM at $\eta_A=0.1,0.7$ and $1.0$ compared to BB84 for the parameters based on GYS.\label{fig:f3}}
\end{figure}
\begin{figure}[h!]
	\centering
			\includegraphics[angle=270,width=0.6\textwidth]{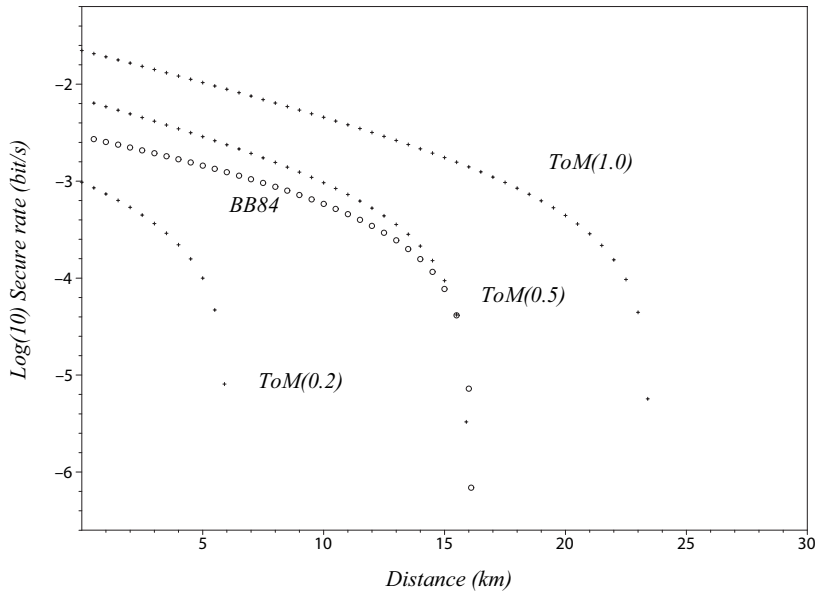}
		\caption{The above exhibits the secure key gain for ToM at $\eta_A=0.2,0.5$ and $1.0$ compared to BB84 for the parameters based on KTH.\label{fig:f4}}
\end{figure}
In consideration of the above parameters, the efficiency of Alice's equipments seems to require more attention for improvement when compared to the channel transmitivity.

\newpage
\section{Actual Photon Contributions \& Absence of PNS}
\noindent In this section we consider the case where Alice and Bob can actually determine the number of photons in a pulse contributing to the key (in an infinite decoy case for example; a simple study for decoy implementation of LM05 was made in \cite{jis}). In the absence of a PNS attack, the key rate for a BB84 is given by \cite{phd} 
\begin{eqnarray}
R\geq \dfrac{1}{2}\left\{-p_b f(e_b)h(e_b)+p_1[1-\tau(e_1)]\right\}
\label{bb84}
\end{eqnarray}
where $p_i=y_i\mu^i/i!$, $y_i=p^{dark}+\eta_i-p^{dark}\eta_i$ and $\eta_i=1-\left(1-\eta\right)^i$, $e_i=\left(e_0p^{dark}+e_{det}\eta_i\right)/y_i$ and $i$ refers to the number of photons in a pulse of concern.
As for ToM, we write
\begin{eqnarray}
R\geq \dfrac{1}{2}\left\{-p_b f(e_t)h(e_t)+p_1[1-\tau_{t}(e_1)]+p_2[1-\tau(e_2)]\right\}
\label{tom}
\end{eqnarray}
which, not surprisingly, is somewhat similar to the decoy formula for LM05 \cite{jis}. We note at passing that the calculations for eq.(17) include the relevant $\eta$ for two way channels. 
\begin{figure}[h]
	\centering
			\includegraphics[angle=270,width=0.6\textwidth]{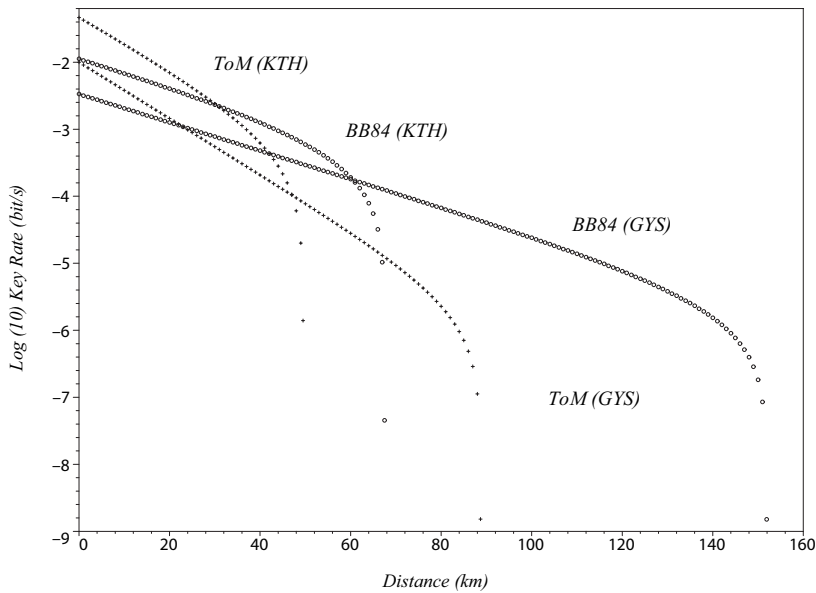}
	\caption{The above exhibits the secure key gain for ToM and BB84 based on eq.(16) and eq. (17) for GYS and KTH parameters.\label{fig:f5}}
\end{figure} 

The plot in Figure \ref{fig:f5} for the key rates of BB84 and ToM using eq. (16) and eq. (17) shows a more favorable picture for the former after about 20 km and 30 km for GYS and KTH parameters respectively. The reason for the advantage for ToM at shorter distances would mainly be due to double photon contributions. This is indeed a reminiscence of the SARG04 protocol which suggests a better key rate compared to BB84 though under decoy implementations (in the absence of PNS) performs otherwise \cite{scarev}. Relating the argument from \cite{scarev} for SARG04, we can say that ToM performs better than BB84 under PNS considerations. It would be interesting to imagine that two way QKD schemes may just be a natural protocol against PNS attacks while possibly maintaining a higher level of robustness against noise. 

\section{Two Way Quantum Cryptographic Scheme \& Continuous Variables}
\noindent In this short section, we should like to brief on a two way quantum cryptographic scheme in the framework of continuous variables (CV) as proposed in \cite{pir}. 
Although originally described in the language of finite dimensional Hilbert space, QKD has been explored in infinite dimensional Hilbert spaces in the framework of CV \cite{mh,tc,gr,gr2,we,lance}.  Encoding is done by amplitude modulation of coherent states with an independent pair of Gaussian variables, $Q$ and $P$ and decoding is done by measurements of the quadratures. Earlier proposals are, in philosophy not unlike its finite dimensional counterpart, i.e. as prepare and measure schemes. 

In \cite{pir}, a two way quantum cryptographic scheme in the continuous variable framework was proposed. This protocol sees Bob sending to Alice half an EPR pair, $B_2$ while keeping one of the modes, $B_1$ to himself. Alice would perform a Gaussian modulation by adding a stochastic amplitude to encode information before resending it to Bob who would then, together with $B_1$  resort to either a homodyne detection (disjoint measurement of either $Q$ or $P$ quadratures) or a heterodyne detection (joint measurements of $Q$ and $P$  quadratures). In analyzing the security of such a protocol, it was shown that Eve necessarily attacks both paths (channels) and given one mode Gaussian attacks, such a two way protocol provides improvement of the security threshold over one way protocols. This may be interpreted as having secure performance in a pair of channels of which individually would be too noisy for one way QKD. This effect is referred to by the authors as superadditivity and represents a definite advantage over one way QKD. Their proposal becomes essentially complete in the formulation of a hybrid protocol where Alice randomly chooses between a two way and a one way protocol, with the latter being identical to a prepare and measure scheme. The motivation for such a construction is as a measure against the most general collective attack on a two way scheme where Eve may perform an attack engendering correlation between both paths. Such correlations if exist, may be detected easily by the legitimate parties. Further studies have been made in \cite{pir2} where asymmetric Gaussian attacks between the two paths were considered and it was shown that the superadditive secure threshold holds. In \cite{pir3}, a specific class of individual attack using combinations of Gaussian cloning machines on one of the protocols proposed in \cite{pir} was analyzed.

\section{Conclusion}

\noindent In order to study what we believe to be the essentials of a two way QKD scheme (LM05) rigorously, we proposed a simple toy model, ToM. Beginning with an independent attack using a two dimensional ancilla, we proved a simple but relevant theorem and argue what the best independent attack should be for Eve. While an IR attack proves to be the more reasonable choice for concern, given an implementation with an imperfect photon source, we had had to resort to consider Eve's Renyi information gain for the double photon contribution case instead. We believe our use of key rate formulas and the like are more rigorously justified. We proceed to ascertain the secure key rates using BB84 based experimental parameters of GYS and KTH. 

Considerations for a completely efficient Alice shed a favorable light on ToM in comparisons against BB84 for both sets of parameters. This highlights an important feature of two way schemes, i.e. the inclusion of double photon contributions play a very significant role in key generation and allows for a higher key rate compared to BB84 despite the relatively extreme toll of lossy channels on ToM. We derived a sufficient condition for ToM's advantage over BB84 and later noted for the parameters above, the ultimate culprit seems to be the transmitivity of Alice's equipments. 

On the other hand, in the absence of a PNS attack, Alice's ability to ascertain the contributing photon yields results in ToM displaying lower key rates compared to BB84 after a certain distance. This is quickly compared to the case for SARG04, a protocol designed with the intention of combating the PNS. 

We hope to think that the toy model here would pronounce the interesting features of an actual two way QKD scheme like LM05. We believe the key rate formula for ToM is more pessimistic than that used or ought to be for a proper LM05 \cite{marco} especially given the appendage of the half factor term as well as the exclusion of three photon term. More importantly, in ToM, Eve's attack tends to leave her with equal amount of information about Alice encoding as well as Bob's measured state. This is not the case for LM05 which sees an asymmetry between the two, hence allowing for Alice and Bob to engage in a reverse reconciliation procedure \cite{arxivrev} that should in principle decrease the amount of information to be discarded in privacy amplification. A note worthy of mention is that the theory of two way QKD protocols boasts of higher elements of robustness against BB84 while previous studies in the face of lossy channels on the other hand have understandably suggested otherwise. Hence we believe our result should provide a breath of fresh air and effectively proposes for a more serious consideration of two way protocols. In the realm of CV quantum cryptography, promising results have been established in \cite{pir,pir2,pir3} spelling a definite advantage of two way protocols over its one way predecessor.    

\section{Acknowledgment}
One of the authors, J.S.S would like to acknowledge financial support under the project FRGS0510-122 from the Ministry of Higher Education's FRGS grant scheme and the university's Research Management Centre for their assistance and facilities provided.


\begin{thebibliography}{99}
\bibitem{gisin} N. Gisin, G. Ribordy, W. Tittel, and H. Zbinden, Rev. of Mod. Phys. 74, 145 (2002).
\bibitem{bb84} C. H. Bennett and G. Brassard, in Proc. of IEEE Int. Conference on Computers, Systems,
and Signal Processing (Bangalore, India, 1984), pp. 175–179.
\bibitem{lutk} N. Lutkenhaus, Phys. Rev. A 61, 052304 2000
\bibitem{lutj} N. Lutkenhaus, M. Jahma, New Journal of Physics 4 (2002) 44.1–44.9
\bibitem{hut} Huttner B, Imoto N, Gisin N and Mor T 1995 Phys. Rev. A 51 1863 
\bibitem{bech} H. Bechmann-Pasquinucci, Phys. Rev. A 73, 044305 (2006)
\bibitem{bostroem} K. Bostroem, T. Felbinger, Phys. Rev. Lett. 89 (2002) 187902.
\bibitem{cai} Q.-Y. Cai, B.W. Li, Chin. Phys. Lett. 21 (2004) 601.
\bibitem{deng} F.-G. Deng, G.L. Long, Phys. Rev. A 70 (2004) 012311.
\bibitem{LM} M. Lucamarini, S. Mancini, Phys. Rev. Lett. 94 (2005) 140501.
\bibitem{LM05} Alessandro Cere, Marco Lucamarini, Giovanni Di Giuseppe, and Paolo Tombesi, Phys. Rev. Lett. 96, 200501 (2006)
\bibitem{chiri} G. Chiribella, G. M. D'Ariano, P. Perinotti Phys. Rev. Lett. 101, 180504 (2008)
\bibitem{marco} Marco Lucamarini, Alessandro Cere, Giovanni Di Giuseppe, Stefano Mancini, David Vitali and Paolo Tombesi, Open Sys. \& Information Dyn. (2007) 14:169–178
\bibitem{jes} J. S. Shaari, M. Lucamarini, and M. R. B. Wahiddin, Phys. Lett. A 358, 85 (2006). 
\bibitem{arxivrev} Marco Lucamarini and Stefano Mancini, arXiv:1004.0157v1 [quant-ph] 1 Apr 2010
\bibitem{jis} J. S. Shaari, Iskandar Bahari and Sellami Ali, arXiv:1006.1693v1 [quant-ph] 9 Jun 2010
\bibitem{renyi} A. Renyi, in Proc. 4th Berkeley Symp. on Mathematical Statistics and Probability, Vol. 1, 1961, pp. 547-561.
\bibitem{how} H. E. Brandt, Journal of Physics: Conference Series 70 (2007) 012005
\bibitem{X} Xiongfeng Ma, Bing Qi, Yi Zhao, and Hoi-Kwong Lo,arXiv:0503005v5 [quant-ph] 10 May 2005
\bibitem{gys} C. Gobby, Z. L. Yuan, and A. J. Shields, Applied Physics Letters, Volume 84, Issue 19, pp. 3762-3764, (2004).
\bibitem{kth} M. Bourennane, F. Gibson, A. Karlsson, A. Hening, P.Jonsson, T. Tsegaye, D. Ljunggren, and
E. Sundberg, Opt. Express 4, 383 (1999).
\bibitem{phd} Xiongfeng Ma, arXiv:0808.1385v1 [quant-ph] 10 Aug 2008
\bibitem{scarev} V. Scarani, H. Bechmann-Pasquinucci, N. J. Cerf, N. Lutkenhaus, M. Peev, Rev. OF Mod. Phys. 81, (2009)
\bibitem{pir} S. Pirandola, S. Mancini, S. Lloyd, S. L. Braunstein, Nature Physics 4, 726 (2008)
\bibitem{mh} M. Hillery, Phys. Rev. A 61, 022309 (2000).
\bibitem{tc} T. C. Ralph, Phys. Rev. A 61, 010303(R) (2000).
\bibitem{gr} F. Grosshans and Ph. Grangier, Phys. Rev. Lett. 88, 057902 (2002).
\bibitem{gr2} F. Grosshans et al., Nature 421, 238-241 (2003).
\bibitem{we} C. Weedbrook et al., Phys. Rev. Lett. 93, 170504 (2004).
\bibitem{lance} A. M. Lance et al., Phys. Rev. Lett. 95, 180503 (2005).
\bibitem{pir2} S. Pirandola, S. Mancini, S. Lloyd, S. L. Braunstein,  Proc. SPIE, Vol. 7092, 709215 (2008).
\bibitem{pir3} S. Pirandola, S. Mancini, S. Lloyd, S. L. Braunstein, Proc. of ICQNM 2009 (IEEE, Los Alamitos, California, 2009), pp. 38-41

\end{thebibliography}
\end{document}